\documentclass{sn-jnl}
\usepackage{manyfoot}%
\usepackage{fullpage}
\usepackage{algorithm}%
\usepackage{algorithmicx}%
\usepackage{algpseudocode}%

\floatname{algorithm}{Algorithm}

\usepackage{graphicx}
\usepackage{amsmath}
\usepackage{psfrag}
\usepackage{color,amssymb}
\usepackage{amsthm}
\theoremstyle{plain}%
\usepackage{thmtools, thm-restate}

\begin{document}

\title[Article Title]{Refined Graph Encoder Embedding via Self-Training and Latent Community Recovery}

\author*[1]{\fnm{Cencheng} \sur{Shen}}\email{shenc@udel.edu}

\affil[1]{\orgname{University of Delaware}, \orgaddress{\city{Newark}, \postcode{19716}, \state{DE}, \country{US}}}

\author[2]{\fnm{Jonathan} \sur{Larson}}

\affil[2]{\orgname{Microsoft Research}, \orgaddress{\city{Redmond}, \postcode{98052}, \state{WA}, \country{US}}}

\author[2]{\fnm{Ha} \sur{Trinh}}


\author[3]{\fnm{Carey E.} \sur{Priebe}}

\affil[3]{\orgname{Johns Hopkins University}, \orgaddress{\city{Baltimore}, \postcode{21218}, \state{MD}, \country{US}}}

\abstract{This paper introduces a refined graph encoder embedding method, enhancing the original graph encoder embedding through linear transformation, self-training, and hidden community recovery within observed communities. We provide the theoretical rationale for the refinement procedure, demonstrating how and why our proposed method can effectively identify useful hidden communities under stochastic block models. Furthermore, we show how the refinement method leads to improved vertex embedding and better decision boundaries for subsequent vertex classification. The efficacy of our approach is validated through numerical experiments, which exhibit clear advantages in identifying meaningful latent communities and improved vertex classification across a collection of simulated and real-world graph data.}


\maketitle

\section{Introduction}
Graph data has surged in popularity, serving as an ideal data structure for capturing interactions across diverse domains, including social networks, citation systems, communication networks, and physical and biological systems \cite{GirvanNewman2002, newman2003structure, barabasi2004network, boccaletti2006complex, VarchneyEtAl2011, ugander2011anatomy}. This rise is driven by the increasing availability of public graph datasets \cite{snapnets,nr,OGBData} and a growing interest in graph learning techniques, such as graph convolutional networks (GCNs) \cite{kipf2017semi, Wu2019ACS, xu2019powerful}, node2vec \cite{grover2016node2vec}, and information network embedding \cite{tang2015line}. These methods have enabled numerous applications, such as knowledge graph embeddings \cite{dettmers2018convolutional}, graph recommendation \cite{he2020lightgcn}, label propagation \cite{wang2022combining}, cross-media retrieval \cite{yao2024efficient}, and contrastive learning for clustering \cite{zhao2024graph}, among many others.

Many graph analysis methods can be viewed as graph embedding techniques, which project graph-structured data into low-dimensional spaces while preserving both structural and semantic properties. Among them, spectral embedding holds a foundational role in graph data analysis. It maps graph data into a $d$-dimensional space using either the graph adjacency matrix or the graph Laplacian \cite{RoheEtAl2011,SussmanEtAl2012, Priebe2019}. Spectral embedding is particularly notable for its strong theoretical properties, as its vertex embeddings converge to the underlying latent positions under popular random graph models \cite{SussmanTangPriebe2014, JMLR:v18:17-448, Patrick2022b}. Consequently, spectral embedding provides a versatile and theoretically sound graph learning technique, with many works spanning vertex classification \cite{TangSussmanPriebe2013, mehta2021neuronal}, community detection \cite{mu2022spectral, gallagher2023spectral}, vertex nomination \cite{zheng2022vertex}, and the analysis of multiple graphs and time-series \cite{arroyo2021inference, Patrick2022a}.

However, the scalability of spectral embedding has been a major bottleneck due to its use of singular value decomposition (SVD), which can be time-consuming for moderate to large graphs. To address this bottleneck, a recent method called one-hot graph encoder embedding \cite{GEE1} has been proposed, offering a supervised extension of spectral embedding. By incorporating vertex labels when available, this approach significantly improves computational efficiency. Furthermore, the graph encoder embedding often outperforms spectral embedding in finite-sample performance across diverse applications, such as large-scale dynamic graph analysis \cite{GEEDynamics}, general graphs including distance and kernel matrices \cite{GEEDistance, EncoderKernel}, and computing graph covariance and correlation \cite{GEECorr,GraphCovLLM}.

In this paper, we propose an enhanced version of the graph encoder embedding that incorporates self-training and latent community recovery. The refined graph encoder embedding can detect hidden sub-communities within the given labels. It has the capability to detect all sub-communities if desired, or to selectively detect latent communities that can benefit subsequent vertex classification and prevent over-refinement. Using stochastic block models with varying parameter choices, we provide illustrative examples to show what the refined method can detect, the meaning of the recovered communities in context, and why detection shall be restricted for the benefit of vertex classification. Theoretical rationale is provided, and the proposed method is validated through simulations and a wide range of real graph data. 
All experiments are carried out on a local desktop with MATLAB 2024a, Windows 10, an Intel 16-core CPU, and 64GB of memory. The appendix includes algorithm pseudocode, theorem proofs, and additional numerical experiments. The code and data are available on GitHub\footnote{\url{https://github.com/cshen6/GraphEmd}}. 

\section{Review}
In this section, we briefly review the stochastic block model (SBM), a classical random graph model. This is followed by an overview of the adjacency spectral embedding and the original graph encoder embedding.

\subsection{Graph Adjacency and Stochastic Block Models}

A graph consists of a set of vertices $\{v_i, i=1,\ldots,n\}$ and a set of edges $\{e_j, j=1,\ldots,s\}$, which can be succinctly represented by an $n \times n$ adjacency matrix $\mathbf{A}$. In this matrix, $\mathbf{A}(i,j)=0$ indicates the absence of an edge between vertex $i$ and $j$, while $\mathbf{A}(i,j)=1$ indicates the existence of an edge. The adjacency matrix can also be weighted to reflect a weighted graph and, more generally, can represent any similarity or dissimilarity matrix, such as pairwise distance or kernel matrices.

Under SBM, each vertex $i$ is assigned a label $\mathbf{Y}(i) \in \{1,\ldots, K\}$. The probability of an edge between a vertex from class $k$ and a vertex from class $l$ is determined by a block probability matrix $\mathbf{B}=[\mathbf{B}(k,l)] \in [0,1]^{K \times K}$. For any $i \neq j$, it holds that
\begin{align*}
\mathbf{A}(i,j) &\sim \operatorname{Bernoulli}(\mathbf{B}(\mathbf{Y}(i), \mathbf{Y}(j))).
\end{align*}
To generate an undirected graph, simply set $\mathbf{A}(j,i) = \mathbf{A}(i,j)$ for all $i<j$. More details on SBM can be found in \cite{HollandEtAl1983, SnijdersNowicki1997, airoldi2008mixed, KarrerNewman2011}. 

Moreover, the degree-corrected stochastic block model is an extension of the SBM that accounts for the sparsity observed in real graphs \cite{ZhaoLevinaZhu2012,yan2012model}. It assigns a degree parameter $\theta_i \in [0,1]$ to each vertex $i$. Given the degrees, each edge from vertex $i$ to another vertex $j$ is independently generated as follows: 
\begin{align*}
\mathbf{A}(i,j) \sim \operatorname{Bernoulli}(\theta_i \theta_j \mathbf{B}(\mathbf{Y}(i), \mathbf{Y}(j))).
\end{align*}

\subsection{Spectral Embedding and Encoder Embedding}
\label{rev1}
Spectral embedding is a classical method in machine learning \cite{Jordan2002,BelkinNiyogi2003}. When applied to a graph adjacency matrix $\mathbf{A} \in \mathbb{R}^{n \times n}$, it utilizes the singular value decomposition (SVD): 
\begin{align*}
\mathbf{A} = \mathbf{U}\mathbf{S}\mathbf{V}^{T}.
\end{align*}
Let $\mathbf{S}_{d}$ be the first $d \times d$ submatrix of $\mathbf{S}$, and $\mathbf{V}_{d}$ be the first $n \times d$ submatrix of $\mathbf{V}$. The adjacency spectral embedding (ASE) is computed as: 
\begin{align*}
\mathbf{Z}^{ASE}=\mathbf{V}_{d}\mathbf{S}_{d}^{0.5} \in \mathbb{R}^{n \times d},
\end{align*}
where the $i$th row of $\mathbf{Z}^{ASE}$ represents the vertex embedding of vertex $i$. The Laplacian spectral embedding (LSE) follows the same formulation, except the adjacency matrix $\mathbf{A}$ is replaced by the corresponding graph Laplacian $\mathbf{L}$. For detailed discussions on graph spectral embedding and its theoretical properties under random graph models, see \cite{RoheEtAl2011, SussmanEtAl2012, Priebe2019}.

The graph encoder embedding (GEE) is a recent method that can be viewed as a supervised extension of spectral embedding, incorporating label information via a label vector $\mathbf{Y} \in [0,1,\ldots,K]^{n}$. Here, $K$ denotes the total number of classes, and $\mathbf{Y}(i) = 0$ indicates that the label for vertex $i$ is unknown. The method proceeds as follows:
\begin{itemize}
\item \textbf{Class Observation Count}: Compute the number of known labels for each class: \begin{align*}
n_k = \sum_{i=1}^{n} 1(\mathbf{Y}(i)=k)
\end{align*}
for $k = 1, \ldots, K$.
\item \textbf{Normalized One-Hot Encoding Matrix}: Construct a normalized one-hot encoding matrix $\mathbf{W} \in [0,1]^{n \times K}$. For each vertex $i = 1, \ldots, n$, set: 
\begin{align*}
\mathbf{W}(i, k) = 1 / n_k
\end{align*} 
if $\mathbf{Y}(i) = k$, and $0$ otherwise.
\item \textbf{Embedding Computation}: Compute the embedding through a simple matrix multiplication: 
\begin{align*}
\mathbf{Z}=\mathbf{A} \mathbf{W} \in [0,1]^{n \times K}.
\end{align*} 
where $\mathbf{Z}(i, :)$ represents the embedding of vertex $i$.
\end{itemize}

Unlike spectral embedding, which derives a projection matrix via singular value decomposition, GEE directly utilizes the label vector to construct the projection matrix $\mathbf{W}$. Moreover, GEE shares similar theoretical properties with spectral embedding, including convergence to the underlying latent positions. For further details on its theoretical properties and iterative techniques for managing unknown labels, refer to \cite{GEE1, GEEClustering, GEEPrincipalCommunity}.

\section{Refined Graph Encoder Embedding}

\subsection{Linear Transformation for Self-Training}
\label{method}
We first outline the steps for applying linear discriminant analysis (LDA) to transform the graph encoder embedding, enabling self-training on the embedding as follows:

\begin{itemize}
\item \textbf{Input}: The graph adjacency matrix $\mathbf{A} \in \{0,1\}^{n \times n}$ and a label vector $\mathbf{Y} \in \{0,1,\ldots,K\}^{n}$, where $1$ to $K$ represent known labels and $0$ is a dummy category for vertices with unknown labels.
\item \textbf{Graph Encoder Embedding}: Compute the graph encoder embedding $\mathbf{Z}$ as described in Section~\ref{rev1}.
\item \textbf{Linear Transformation}: Apply LDA to $\mathbf{Z}$. Specifically, let $\mu_k \in \mathbb{R}^{K}$ denote the class-conditional mean:
\begin{align*}
\mu_k = E(\mathbf{Z}(i,:) | \mathbf{Y}(i)=k)
\end{align*}
for $k=1,\ldots,K$. Denote $\mu=[\mu_1, \mu_2, \ldots, \mu_K] \in \mathbb{R}^{K \times K}$ as the matrix of concatenated means, and let $\Sigma \in \mathbb{R}^{K \times K}$ be the estimated common covariance matrix of $\mathbf{Z}$. Denote $\Sigma^{+}$ as the pseudo-inverse of $\Sigma$, $\vec{n}_K$ as a row vector where the $k$th entry is $n_k$, and $\mbox{diag}(\cdot)$ as an operator that extracts the diagonal terms of a matrix as a row vector. The linearly transformed GEE embedding is computed as: 
\begin{align}
\label{eq1}
\mathbf{Z}_{1} = \mathbf{Z} \Sigma^{+} \mu- (\mbox{diag}(\mu' \Sigma^{+} \mu) - \mbox{log}(\vec{n}_K / n)).
\end{align}
Here, the term after the first minus sign is a row vector, subtracted from each row of $\mathbf{Z} \Sigma^{+} \mu$. Equation~\ref{eq1} simply represents the linear discriminant analysis expressed as a single matrix equation.
\item \textbf{Self-Training}: Compute the updated label vector $\mathbf{Y}_{1}$ and a mismatch indicator vector $\delta_1$ as follows: 
\begin{align*}
\mathbf{Y}_{1}(i) &= \arg\max_{k=1,\ldots,K} \mathbf{Z}_{1}(i,k)\\
\delta_1(i) &= 1\{\mathbf{Y}_{1}(i) \neq \mathbf{Y}(i) \mbox{ and } \mathbf{Y}(i)>0\},
\end{align*}
where $\delta_1(i)$ indicates whether the original label $\mathbf{Y}(i)$ and the self-trained label $\mathbf{Y}{1}(i)$ are mismatched.
\item \textbf{Output}: The linear transformed encoder embedding $\mathbf{Z}_{1} \in \mathbb{R}^{n \times K}$, the updated label vector $\mathbf{Y}_{1} \in \{1,\ldots,K\}^{n}$, and an indicator vector $\delta_1 \in [0,1]^{n}$.
\end{itemize}

The purpose of GEE with LDA is to align the original GEE such that the dimension attaining the maximum value determines the class assignment for each vertex. While one could use a neural network with softmax outputs for this purpose, the graph encoder embedding is approximately normally distributed, as stated in Theorem~\ref{thm1}. Therefore, LDA is a suitable and more computationally efficient choice for estimating conditional probabilities and aligning the embedding. A pseudocode version of the algorithm is provided in appendix algorithm~\ref{alg1}.

\subsection{Refined GEE via Self-Training and Latent Community Recovery}

We present the main method, which outputs a refined graph encoder embedding along with potential latent community information, through iterative self-training and latent community discovery:

\begin{itemize}
\item \textbf{Input}: The graph adjacency matrix $\mathbf{A} \in \{0,1\}^{n \times n}$ and a label vector $\mathbf{Y} \in \{0,1,\ldots,K\}^{n}$, where $1$ to $K$ represent known labels and $0$ is a dummy category for vertices with unknown labels. Additional parameters include the maximum number of refinements for self-training ($\gamma_{Y}$) and latent community discovery ($\gamma_{K}$), set to $5$ by default; stopping criteria $\epsilon \in [0,1]$ and $\epsilon_n \in \mathbb{N}$, set to $0.3$ and $5$ by default.
\item \textbf{Initial Embedding}: Apply LDA to GEE as in Section~\ref{method}, yielding the vertex embedding $\mathbf{Z}_{1}$, the updated label vector $\mathbf{Y}_{1}$, and an indicator vector $\delta_{1}$. 
\item \textbf{Refinement via Self-Training}: Update the vertex embedding using $\mathbf{A}$ and the new label vector $\mathbf{Y}_{1}$. Denote the outputs as $\mathbf{Z}_{2}$, $\mathbf{Y}_{2}$, and $\delta_2$.
\item \textbf{Stopping Criterion}: Check whether $\delta_{1}$ and $\delta_{2}$ are sufficiently close, defined by:
\begin{align}
\label{close}
\sum_{i=1}^{n}\delta_1(i) - max( \sum_{i=1}^{n}\delta_1(i) * \epsilon, \epsilon_n) < \sum_{i=1}^{n}\delta_1(i) \cdot \delta_2(i).
\end{align}
If the condition holds, proceed to the next step. Otherwise, set $\mathbf{Y}_{1} = \mathbf{Y}_{2}$ and repeat the previous step until the iteration limit $\gamma_{Y}$ is reached. After each iteration, concatenate embeddings and labels:
\begin{align*}
\mathbf{Z}_1=[\mathbf{Z}_1,\mathbf{Z}_{2}], \mathbf{Y}=[\mathbf{Y},\mathbf{Y}_1].
\end{align*}
\item \textbf{Refinement via Latent Communities}: Assign mismatched vertices to a new class, update the vertex embedding using $\mathbf{A}$ and $\mathbf{Y}_{1} + \delta_1 * K$. Denote the outputs as $\mathbf{Z}_{2}$, $\mathbf{Y}_{2}$, and $\delta_2$. Repeat this process until the stopping criterion (Equation~\ref{close}) is satisfied or the iteration limit $\gamma_{K}$ is reached. Concatenate intermediate embeddings and labels as in the previous step. 
\item \textbf{Output}: The refined graph encoder embedding $\mathbf{Z}$ and the concatenated label matrix $\mathbf{Y}$. 
\end{itemize}
As an example of the refinement step, suppose the initial label vector is $\mathbf{Y}_1 = [1, 1, 1, 2, 2, 2]$, with $K = 2$. If the 1st and 4th vertices are identified as mismatched according to $\delta_1$, their labels are updated by assigning each to a new group. The updated label vector becomes $\mathbf{Y}_2 = [3, 1, 1, 4, 2, 2]$, where the new label for a mismatched vertex is determined by incrementing its original label by $K$. This refinement process continues iteratively, assigning new labels to mismatched vertices, until the stopping criterion is satisfied.

The iterative self-training phase terminates when the mismatch between the input and self-trained labels becomes insignificant, as determined by the parameters $\epsilon$ and $\epsilon_n$. Specifically, Equation~\ref{close} requires that the updated mismatch indicator $\delta_2$ yields at least a $30\%$ reduction (or a minimum of $5$ indices) compared to the previous mismatch indicator $\delta_1$. Similarly, during latent community discovery, mismatched vertices are iteratively reassigned to new classes until the mismatch indicator no longer shows a significant reduction. At termination, the method outputs the original encoder embedding along with all refined embeddings from both self-training and latent community assignment, concatenated into a single matrix.

The parameters $\gamma_{K}$ and $\gamma_{Y}$ control the number of iterations, while $\epsilon$ and $\epsilon_n$ determine how aggressive the refinement process is by assessing whether the next refinement is sufficiently similar to the previous one. Larger values of $\epsilon$ and $\epsilon_n$ cause the algorithm to stop more easily, while smaller values make the algorithm harder to stop, leading to more aggressive refinement. For example, if the downstream task involves visualizing all hidden communities or detecting outliers, one can set $\epsilon$ and $\epsilon_n$ to their smallest possible values (e.g., $0$) and $\gamma_{K}$ and $\gamma_{Y}$ to large values (e.g., $100$). This ensures that the refinement process continues as long as mismatched indices keep decreasing. Our default parameter choices are designed to be slightly conservative and have performed well across simulations and real data experiments for vertex classification. A pseudocode version of the main method is provided in appendix algorithm~\ref{alg2}, along with additional simulations exploring parameter sensitivity.

\subsection{Running Time Analysis}
The original graph encoder embedding (GEE) has a time complexity of $O(nK+s)$ \cite{GEE1}, where $s$ is the number of edges, making it linear with respect to the number of vertices and edges. Let $K_{M}$ be the largest possible number of refined classes, the refined encoder embedding (R-GEE) in algorithm~\ref{alg2} has a time complexity of $O(n K_{M} + n K_{M}^{2} + s)$, where the quadratic term $K_{M}^{2}$ comes from using linear discriminant. As $K_{M}=\gamma_{K} K$, or at most $5K$ in the default parameter, the method remains linear with respect to the number of vertices and edges.

Figure~\ref{figA1} shows the running time using simulation model 3 with sparse adjacency matrix input (see Section~\ref{sim} for model details), as $n$ increases from $3000$ to $30000$. The average running time and standard deviation are reported based on 10 Monte Carlo replicates, performed on a local desktop running MATLAB 2024a on Windows 10, equipped with an Intel 16-core CPU and 64GB of memory. It is clear that the R-GEE, although slower than the original GEE, is still vastly faster than singular value decomposition, which is the major computational step of spectral embedding. At $n=30000$, the number of edges is about $50$ million; a single SVD into $d=20$ requires about $200$ seconds, while the graph encoder embedding takes $0.4$ second and the refined method takes $1.2$ second. 

\begin{figure}
\centering
\includegraphics[width=0.6\linewidth,trim={0cm 0cm 0cm 0cm},clip]{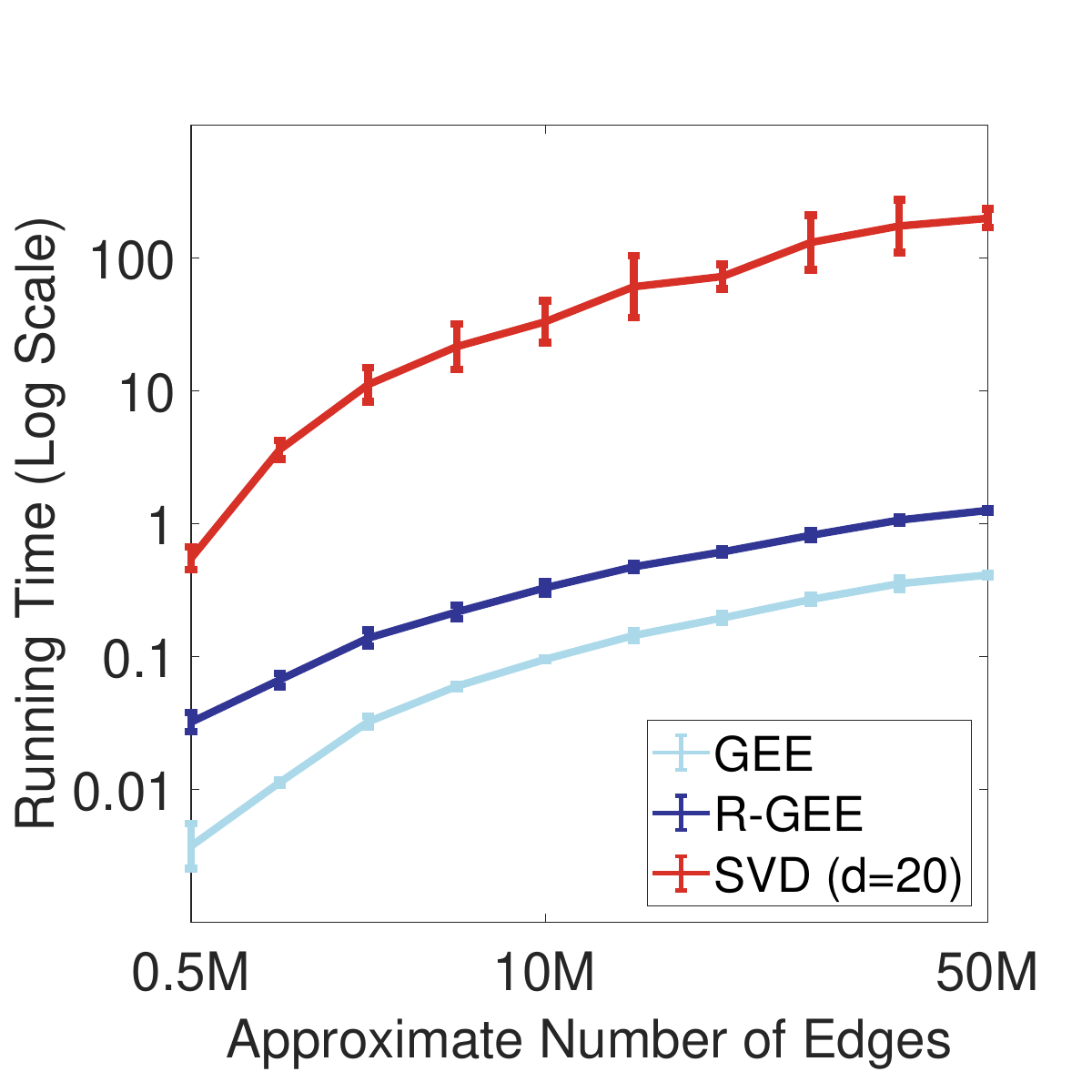}
\caption{This figure shows the running time comparison between GEE, Refined GEE, and SVD. The X-axis represents the approximate number of edges, and the Y-axis represents the running time on a log-10 scale.}
\label{figA1}
\end{figure}

\section{Theoretical Rationale}

In this section, we provide the theoretical rationale for the self-training by linear transformation and the latent community assignment, explaining why the proposed method works and when community refinement helps improve embedding quality. All proofs are provided in the appendix.

\begin{restatable}{theorem}{clt}
\label{thm1}
The graph encoder embedding is asymptotically normally distributed under SBM. Specifically, as $n$ increases, for a given $i$th vertex of class $y$, it holds that
\begin{align*}
\mbox{diag}(\vec{n}_K)^{0.5} \cdot (\mathbf{Z}(i,:) - \mu_{y})  \stackrel{n}{\rightarrow}  \mathcal{N}(0,\Sigma_{y}).
\end{align*}
The expectation and covariance are: $\mu_{y}=\mathbf{B}(y,:)$ and $\Sigma_{y}(k,k)=\mathbf{B}(y,k)(1-\mathbf{B}(y,k))$. Assuming $\Sigma_{y}$ is the same across all $y \in [1,K]$, the transformation in Equation~\ref{eq1} satisfies
\begin{align*}
\mathbf{Z}_{1}(i,k)\stackrel{n}{\rightarrow}  Prob(Y=k| X=\mathbf{Z}(i,:)).
\end{align*}
\end{restatable}

Theorem~\ref{thm1} shows that the original graph encoder embedding is asymptotically normally distributed. As a result, the proposed linear transformation approximates the conditional probability, making it an appropriate choice for subsequent self-training.

\begin{restatable}{theorem}{refine}
\label{thm2}
Suppose the graph is distributed as the stochastic block model with block probability $\mathbf{B} \in \mathbb{R}^{K \times K}$ and observed label vector $\mathbf{Y} \in [1,\ldots,K]$. Then for any two vertices $i, j$, the encoder embedding $\mathbf{Z}$ using observed labels satisfies:
\begin{align*}
\| \mathbf{Z}(i,:) - \mathbf{Z}(j,:) \|_2 - \|\mathbf{B}(\mathbf{Y}(i),:) - \mathbf{B}(\mathbf{Y}(j),:)\|_2 \stackrel{n}{\rightarrow} 0
\end{align*}

Suppose the same graph can be viewed as a realization of a latent stochastic block model with $\mathbf{B}_{0} \in \mathbb{R}^{K_0 \times K_0}$ and a latent label vector $\mathbf{Y}_{0} \in [1,\ldots,K_0]$ where $K_0 >K$. The for the same two vertices $i,j$, the resulting encoder embedding $\mathbf{Z}_0$ using the latent labels satisfies:
\begin{align*}
\| \mathbf{Z}_0(i,:) - \mathbf{Z}_0(j,:) \|_2 - \|\mathbf{B}_0(\mathbf{Y}_{0}(i),:) - \mathbf{B}_{0}(\mathbf{Y}_{0}(j),:)\|_2 \stackrel{n}{\rightarrow} 0
\end{align*}
\end{restatable}

Theorem~\ref{thm2} suggests that when comparing the encoder embedding using observed labels versus the encoder embedding using latent labels, the margin of separation fully depends on the block probability vector between the observed model and the latent model. This means that, from a margin separation perspective, the refined encoder embedding using latent communities could perform better or worse than the original encoder embedding using observed communities. Therefore, for the refined GEE to improve over the original GEE, it needs to properly decide whether to refine the given labels or not. Moreover, it is important to concatenate the embedding in each refinement, because the concatenated embedding retains previous embedding information and is more robust against slight over-refinement.

Under the stochastic block model, this theorem can help verify whether the latent community leads to an improvement or deterioration in the margin of separation over the observed community. Note that the theorem focuses on asymptotic behavior. In finite-sample performance, the embedding variance certainly plays a role in the decision boundary. In this paper, we only considered the mean difference to illustrate the key idea, for simplicity of presentation and to avoid overly complicating mathematical expressions. This is because the variance is generally similar across the groups and bounded above by $0.25$ in SBM.

\section{Simulations}
\label{sim}
We start with three stochastic block models, each serving as a representative case, and use Theorem~\ref{thm2} to verify whether the latent community leads to better embedding separation among groups. We then use embedding visualization, vertex classification, and precision/recall metrics to verify the results and assess the effectiveness of the refined algorithm.

\subsection{Model Parameters}

\subsubsection*{Simulated Graph 1}
For each vertex, we let the latent communities as $Y_{0} \in \{1,2,3,4\}$ with probability $0.25$ each, set the latent block probability matrix as
\begin{align*}
\mathbf{B}_{0}=\begin{bmatrix}
0.5, 0.2, 0.1, 0.1 \\
0.2, 0.2, 0.1, 0.1 \\
0.1, 0.1, 0.2, 0.2 \\
0.1, 0.1, 0.2, 0.5
\end{bmatrix},
\end{align*}
and generate the degree parameter by $\theta_i \stackrel{i.i.d.}{\sim} Uniform(0.1,1)$. Next, we set the observed communities as $Y=1$ if $Y_{0} \in \{1,2\}$, and $Y=2$ if $Y_{0} \in \{3,4\}$. Namely, the first two latent communities are observed as one group, while the last two latent communities are observed as another group. Therefore, the observed block probability matrix can be computed as
\begin{align*}
\mathbf{B}=\begin{bmatrix}
0.275, 0.1 \\
0.1, 0.275 
\end{bmatrix}.
\end{align*}
Now we use Theorem~\ref{thm2} to check the margin of separation. When using the latent labels, the margin of separation between classes $2$ and $3$ equals $\|(0.2, 0.2, 0.1, 0.1) - (0.1, 0.1, 0.2, 0.2)\|=0.2$. When using the observed labels, the difference is $\|(0.275, 0.1) - (0.1, 0.275)\| =0.25$. Therefore, using the observed labels actually provides a larger margin of separation between these vertices. Note that if we consider the separation between classes $1$ and $4$, then the latent communities are better; however, those two latent groups are less important than the separation between latent classes $2$ and $3$.

\subsubsection*{Simulated Graph 2}

The latent communities and block probability matrix are exactly the same as in simulated graph 1,  except the observed communities are set up as follows: $Y=1$ if $Y_{0} \in \{1,3\}$, and $Y=2$ if $Y_{0} \in \{2,4\}$. As a result, the observed block probability matrix can be computed as
\begin{align*}
\mathbf{B}=\begin{bmatrix}
0.225, 0.15 \\
0.15, 0.225 
\end{bmatrix}.
\end{align*}
In this case, the latent communities have a margin of $0.2$ between latent class 2 and 3, which becomes $0.11$ when using observed communities. Therefore, this simulation provides an example where label refinement is necessary and significantly improves the embedding quality.

\subsubsection*{Simulated Graph 3}
In this simulation, we set the latent communities as $Y_{0} \in \{1,2,3,4,5\}$ with probability $0.2$ each, and set the latent block probability matrix as
\begin{align*}
\mathbf{B}_{0}=\begin{bmatrix}
0.5, 0.2, 0.2, 0.1, 0.1 \\
0.1, 0.2, 0.1, 0.2, 0.1 \\
0.1, 0.1, 0.2, 0.1, 0.2 \\
0.1, 0.2, 0.1, 0.5, 0.1 \\
0.1, 0.1, 0.2, 0.1, 0.5
\end{bmatrix}.
\end{align*}
The observed communities are: $Y=1$ if $Y_{0}\in \{1,2,3\}$; $Y=2$ if $Y_{0}=4$; and $Y=3$ if $Y_{0}=5$. Then the observed block probability matrix can be computed as
\begin{align*}
\mathbf{B}=\begin{bmatrix}
0.178, 0.133, 0.133\\
0.133, 0.500, 0.100 \\
0.133, 0.100, 0.500 
\end{bmatrix}
\end{align*}
This simulation is somewhat similar to simulated graph 1 but presents an interesting mixed situation where some decision boundaries are improved using latent groups, while others are worse. For example, the difference between vertices in latent group 2 and 4 is $0.3$ using latent labels, which is enlarged to $0.37$ using observed labels. However, class 1 and 4 are separated by $0.5745$ using latent labels, which is reduced to $0.37$ using observed labels, and similarly for class $1$ versus $5$, or class $4$ versus $5$. 

\subsection{Latent Community Recovery}

Figure~\ref{fig1} shows the graph connectivity for simulated graphs 1 and 2, with vertices colored by latent community (left), observed community (center), and GEE-refined community (right), which uses the label output via one label refinement of R-GEE ($\gamma_{K}=1$ and $\gamma_{Y}=0$).

Since the latent communities do not improve the embedding separation for simulated graph 1, we expect R-GEE to largely ignore the latent communities. This is indeed the case in the top row of Figure~\ref{fig1}, where the refinement only highlights a few vertices in the middle, and most vertices remain in their observed groups.

The situation is different for simulated graph 2, where the latent communities significantly improve the embedding quality and decision boundary. In this case, R-GEE successfully identifies the latent communities, assigning most vertices in latent communities 2 and 3 to different groups, so the right panel closely matches the left panel in the bottom row of Figure~\ref{fig1}. To maintain a clear and consistent visualization, simulation graph 3 is not shown here, as it merely represents a mixed case between graph 1 and graph 2.

\begin{figure}
\centering
\includegraphics[width=1.0\linewidth,trim={2cm 0cm 2cm -0.2cm},clip]{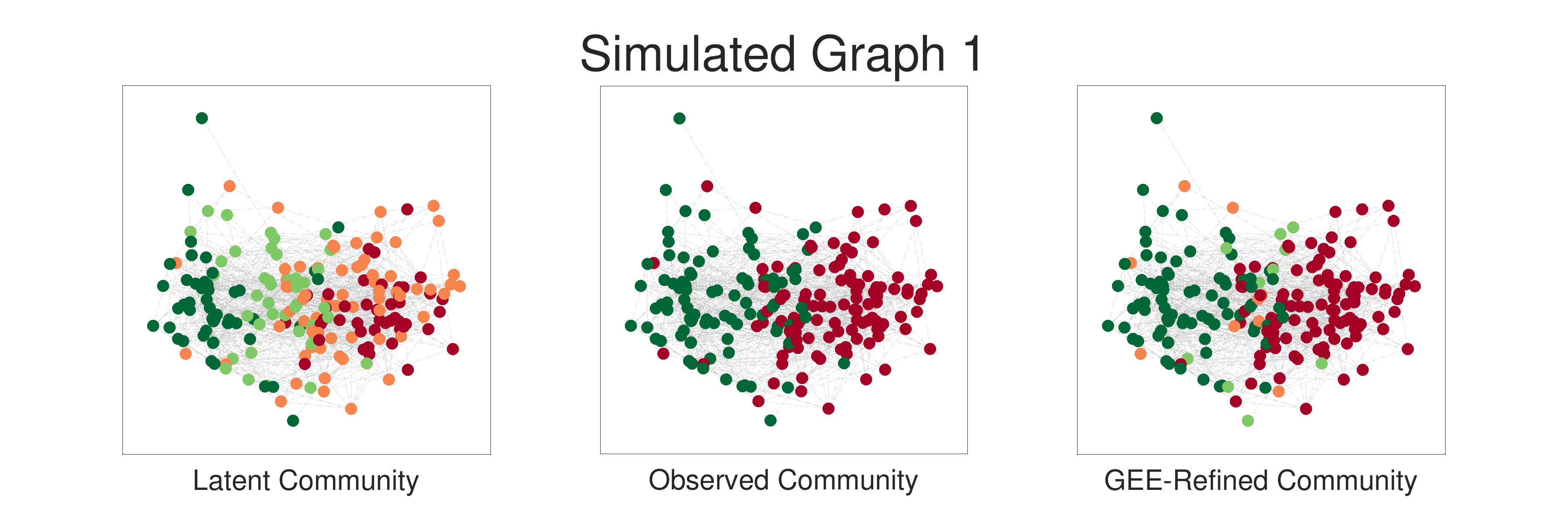}
\includegraphics[width=1.0\linewidth,trim={2cm 0cm 2cm -0.2cm},clip]{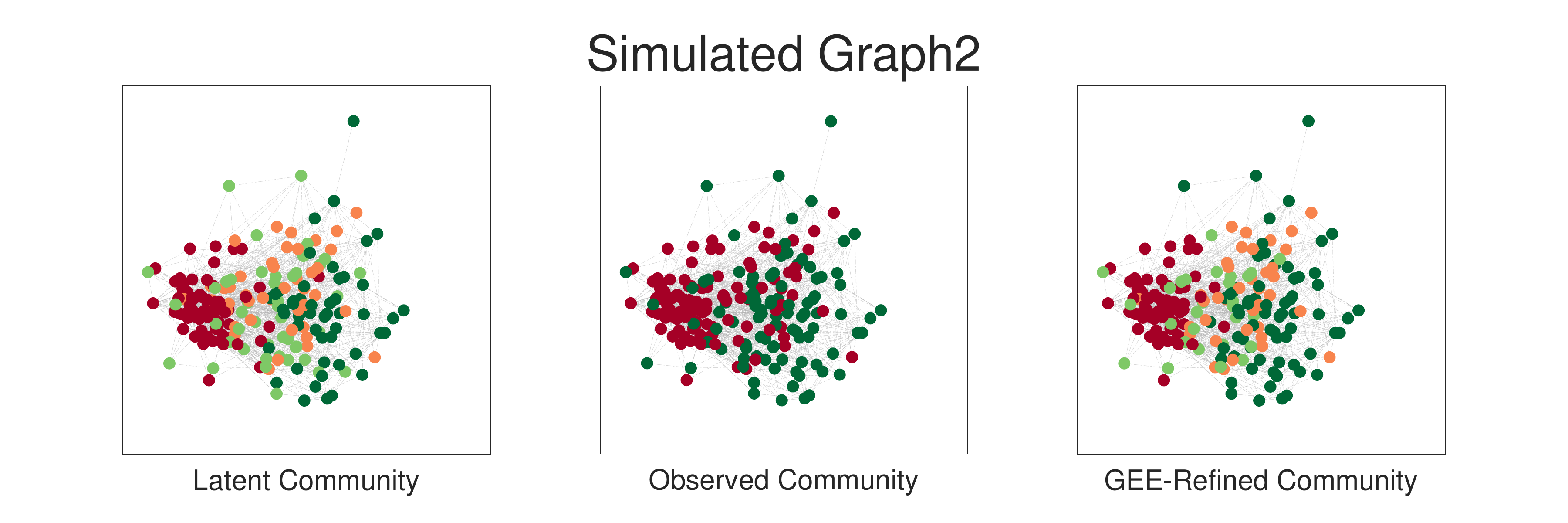}
\caption{This figure visualizes the graph using latent labels (left panel), observed labels (center panel), and GEE-refined labels after one refinement iteration (right panel). In the left panel, the colors dark green, light green, red, and orange represent the four ground-truth latent communities. In the center panel, the dark green and light green vertices from the left panel are combined into a single observed community, colored dark green, while the red and orange vertices are similarly merged into a single observed community, colored red. In the right panel, the R-GEE algorithm refines the observed labels from the center panel and partially recovers the ground-truth latent communities, with the refined communities again represented by light green and orange.}
\label{fig1}
\end{figure}

\subsection{Vertex Classification Evaluation}

The top row of Figure~\ref{fig2} reports the 10-fold cross-validation for the simulated graphs using 30 replicates. For each replicate, we generate a simulated graph of increasing vertex size, along with the corresponding latent and observed labels.

GEE0 computes the original GEE using the latent community labels $Y_0$. GEE computes the original GEE using the observed community labels $Y$. R-GEE uses the proposed algorithm with default parameters and the observed community labels $Y$ as input. ASE stands for adjacency spectral embedding into $d=20$. Each method is then evaluated via an LDA classifier for the observed labels $Y$. Note that the classification task is always for the observed labels $Y$, and the latent labels $Y_0$ are only used for embedding. Moreover, for all GEE methods, the labels of testing observations are assigned to $0$ prior to the embedding.

For simulated graph 1, GEE0 using the latent labels has the worst classification error, while all other methods perform well and similarly to each other. This result matches the model setting and our previous verification that latent communities yield worse embedding quality. For simulated graph 2, GEE using the observed labels performed the worst, while R-GEE, GEE0, and ASE all performed very well. This is a reversal of simulated graph 1 and also matches the model setting and previous verification that latent communities improve the embedding quality in this case. For simulated graph 3, it is a mixed case where some refinement helps marginally, and GEE0, GEE, and R-GEE all performed relatively well with some small differences.

The bottom row of Figure~\ref{fig2} shows the precision and recall of R-GEE. In simulated graphs 1 and 3, the vertex classification results indicate that latent communities are not important, so while the precision is high (all discovered new communities belong to the true latent communities), the recall is relatively low (many vertices from the latent communities are not discovered). For simulated graph 2, discovering the latent communities is critical, and indeed both precision and recall are very high, showing that R-GEE is performing as intended.

Overall, this figure shows that the proposed algorithm works as designed, recovering latent communities only when they are useful for vertex classification, and retaining excellent embedding quality that is not overly refined, as evidenced by the good classification error that converges to $0$ in every case.

\begin{figure}
\centering
\includegraphics[width=1.0\linewidth,trim={1cm 0cm 2cm -0.2cm},clip]{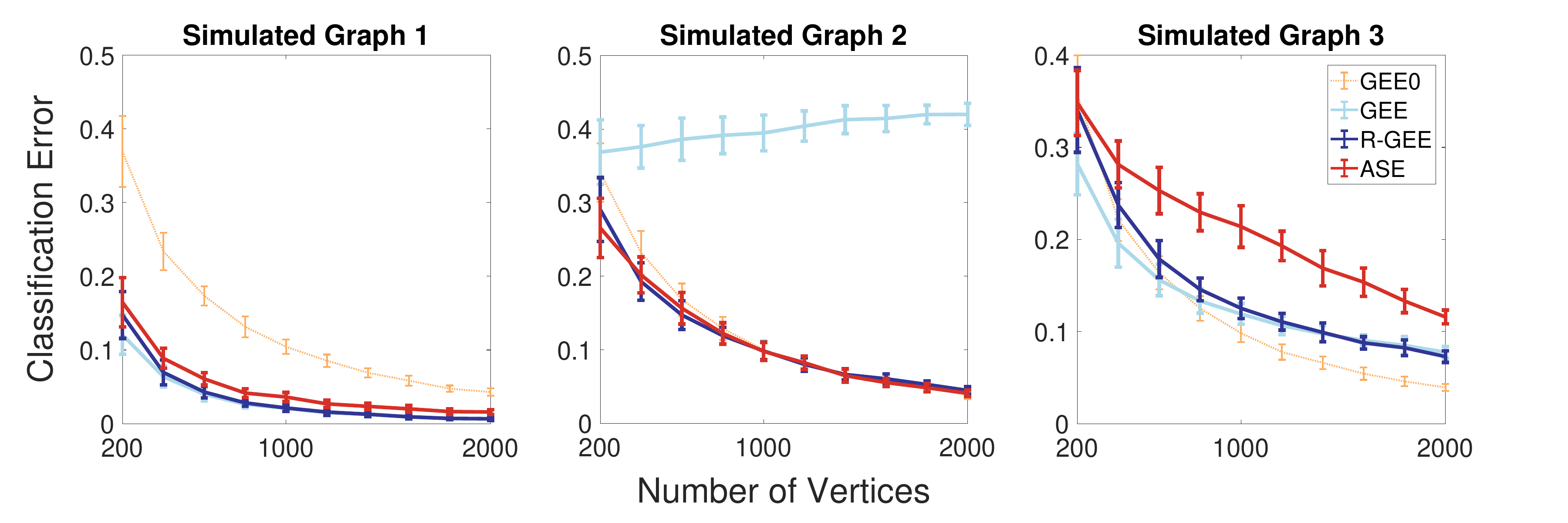}
\includegraphics[width=1.0\linewidth,trim={1cm 0cm 2cm -0.2cm},clip]{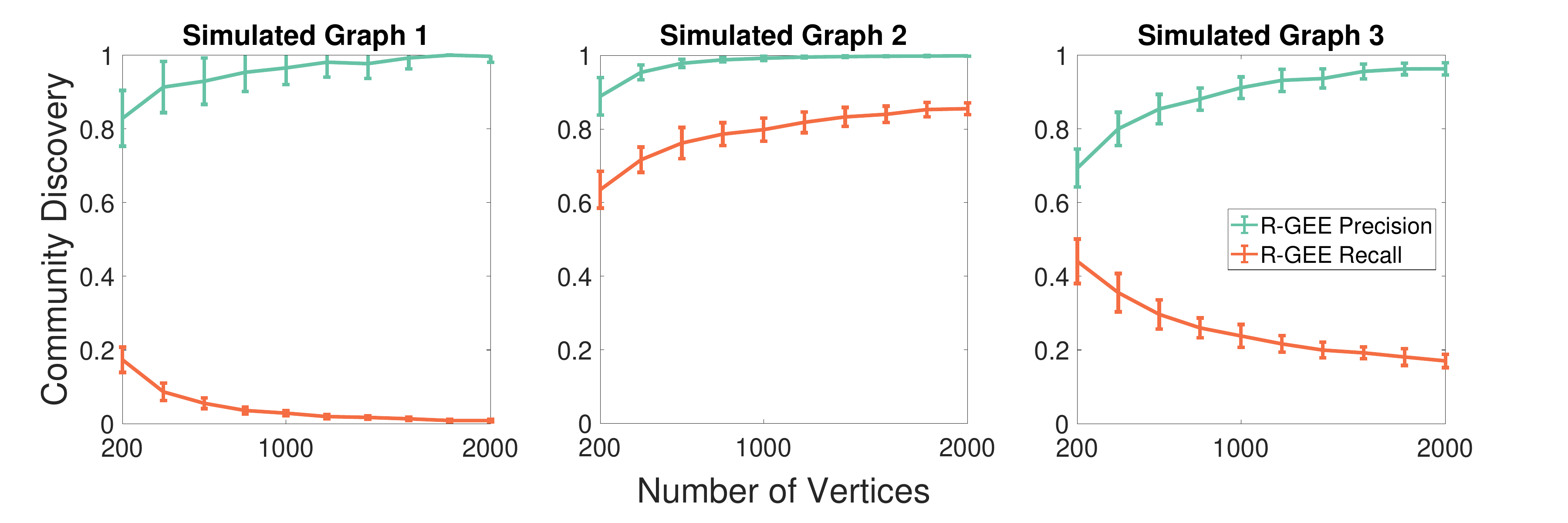}
\caption{The first row of the figure reports the 10-fold cross-validation error and standard deviation for the three simulated graphs, using 30 replicates. The bottom row of the figure reports the precision and recall for refined GEE in recovering the latent communities.}
\label{fig2}
\end{figure}

\section{Real Data Evaluation}
\label{secReal}

We collected a diverse set of real graphs with associated labels from various sources, including the Network Repository\footnote{\url{http://networkrepository.com/}} \cite{nr}, Stanford network data\footnote{\url{https://snap.stanford.edu/}}, and other public graph data. Specifically, we experimented on the adjnoun graph \cite{newman2003structure}, C. elegans neuron data, which provides two binary graphs \cite{PavlovicEtAl2014}, the EU email network \cite{Leskovec2017}, the karate club graph \cite{zachary1977information}, the lastFM Asia social network \cite{lastfm}, the letter graph, the political blogs graph \cite{AdamicGlance2005}, a political retweet graph, the Pubmed Citation network, and a Wikipedia article graph \cite{GCCAJMVA} with four graphs.

\subsection{Vertex Classification}
For a more comprehensive evaluation in the real data experiments, we compared GEE, R-GEE, ASE, LSE, and node2vec. R-GEE used the default parameters; ASE and LSE project into $d=20$ dimensions; node2vec uses the graspy package \cite{Graspy} with default parameters and 128 dimensions. For each dataset and each method, we carried out 10-fold validation and reported the average classification error using LDA, along with one standard deviation, in Table~\ref{table1} with 30 random seeds. Any directed graph was transformed to undirected, and any singleton vertex was removed. Note that unlike the simulated graphs, real graphs do not come with any known latent communities.

Table~\ref{table1} clearly shows that refined GEE is able to preserve or improve the classification error compared to original GEE. In a few cases where it is worse, the difference is only marginal. Moreover, GEE and R-GEE are either the best or very close to the best in terms of classification error across all real data experiments. It should be noted that all methods with parameters could attain better performance if we tuned the parameters for each real dataset, but we chose to use consistent parameter choices throughout the experiments. Therefore, the results reported here should be viewed as a conservative illustration of the proposed method, not the best possible error.


\begin{table}
\renewcommand{\arraystretch}{1.3}
\centering
\small
{\begin{tabular}{|c|c||c|c|c|c|c|}
 \hline
 & (n, K) & R-GEE & GEE & ASE & LSE & N2V\\
\hline
AdjNoun & (112, 2) & $14.9 \pm 2.1$ & $14.6 \pm 2.0$  & $18.7 \pm 1.3$ & $\textbf{14.2} \pm 1.4$ & $45.9 \pm 3.3$ \\
C-Elegans Ac & (253, 3)  & $37.1 \pm 2.7$ & $49.7 \pm 1.5$  & $38.4 \pm 1.3$ & $\textbf{35.3} \pm 1.3$ & $45.3 \pm 2.0$ \\
C-Elegans Ag & (253, 3) & $\textbf{38.0}\pm 2.3$ & $40.8 \pm 1.7$ & $42.3 \pm 1.1$ & $42.6 \pm 1.2$  & $40.2 \pm 2.3$ \\
Coil-RAG & (11687, 100) & $\textbf{19.5} \pm 1.1$ & $\textbf{19.5} \pm 1.1$ & $97.3 \pm 0.1$ & $95.5 \pm 0.1$  & $79.1 \pm 0.2$ \\
Email & (1005, 42) & $29.4 \pm 1.1$ & $33.1 \pm 0.5$ & $44.4 \pm 0.4$ & $88.5 \pm 0.4$  & $\textbf{28.9} \pm 0.4$ \\
Karate & (34, 2) & $\textbf{9.5} \pm 2.3$ & $\textbf{9.5} \pm 2.3$ & $17.4 \pm 4.6$ & $16.6 \pm 4.2$  & $13.8 \pm 2.7$ \\
LastFM & (7624, 18) & $17.5 \pm 0.5$ & $18.1 \pm 0.2$ & $48.9 \pm 0.1$ & $23.1 \pm 0.1$ & $\textbf{14.7} \pm 0.1$ \\
Letter & (10482, 26) & $\textbf{3.2} \pm 0.2$ & $\textbf{3.2} \pm 0.2$ & $89.9 \pm 0.2$ & $88.9 \pm 0.2$  & $74.7 \pm 0.2$ \\
PolBlogs & (1224, 2) & $5.3 \pm 0.3$ & $5.0 \pm 0.3$  & $9.5 \pm 0.2$ & $\textbf{4.9} \pm 0.2$ & $5.1 \pm 0.1$ \\
PolTweet & (1847, 2) & $\textbf{2.6} \pm 0.1$ & $2.7 \pm 0.1$  & $29.8 \pm 0.1$ & $4.6 \pm 0.1$ & $38.8 \pm 0.1$ \\
PubMed & (19716, 3) & $\textbf{20.3} \pm 0.1$  & $20.4 \pm 0.1$ & $37.4 \pm 0.1$ & $34.0 \pm 0.1$ & $58.8 \pm 0.2$ \\
Wiki TE & (1382, 5) & $\textbf{16.0} \pm 0.5$ & $20.4 \pm 0.3$ & $26.2 \pm 0.2$ & $26.6 \pm 0.2$  & n/a\\
Wiki TF & (1382, 5) & $\textbf{15.8} \pm 0.5$ & $20.9 \pm 0.3$ & $27.7 \pm 0.3$ & $27.7 \pm 0.3$  & n/a\\
Wiki GE & (1382, 5) & $\textbf{33.5} \pm 0.9$ & $41.2 \pm 0.5$ & $46.4 \pm 0.4$ & $53.8 \pm 0.4$  & $40.2 \pm 0.5$ \\
Wiki GF & (1382, 5) & $\textbf{42.9} \pm 0.8$ & $50.6 \pm 0.8$ & $47.1 \pm 0.3$ & $56.7 \pm 0.5$  & $46.6 \pm 0.4$ \\
\hline
\end{tabular}
\caption{This table reports the 10-fold vertex classification error and standard deviation for real graphs, using 30 random replicates. All numbers are in percentile. Note that there are two text dissimilarity datasets (Wiki TE and Wiki TF, which are cosine dissimilarity of the underlying articles) where node2vec is not applicable.}
\label{table1}
}
\end{table}

\subsection{Refinement Visualization on Real Data}
\label{appC}
Figure~\ref{fig3} illustrates the community refinement results for two representative cases: the karate club graph and the political blogs graph. R-GEE successfully detects useful hidden communities with just one refinement: for the karate club graph, the algorithm identifies an anomaly vertex that always connects with the other group, and another vertex located at the intersection between the two classes; for the political blogs, our algorithm identifies blogs that are dominantly connected to the other party. Whether these are "swinger" blogs or "imposter" blogs is an issue of practical importance.

While the refined embedding is undoubtedly useful for uncovering potential hidden communities, it is important to note that these latent communities may not always enhance the performance of downstream tasks. For example, in the karate club dataset, Figure~\ref{fig3} clearly identifies two vertices that should form their own community. However, as shown in Table~\ref{table1}, this refinement does not improve vertex classification performance when evaluated against the original class assignments. A similar situation occurs with the political blogs dataset. Although it is visually evident that some vertices should be refined into separate groups, the refinement does not lead to better vertex classification performance. Since the classification benchmark used the original observed labels, we believe this is because the refinement does not always contribute to improved separation of those original labels. 

\begin{figure}
\centering
\includegraphics[width=0.8\linewidth,trim={2cm 0cm 2cm -0.2cm},clip]{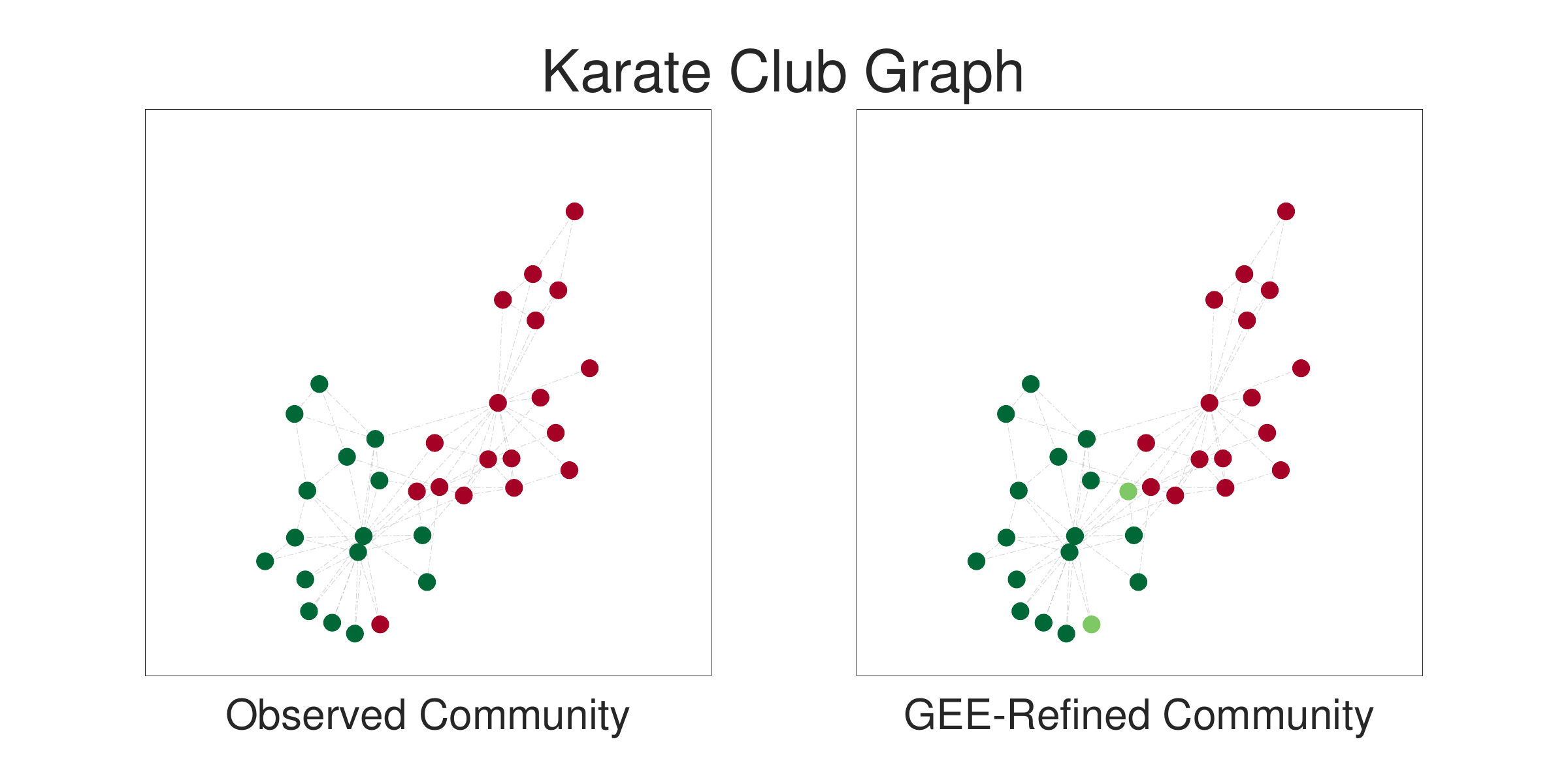}
\includegraphics[width=0.8\linewidth,trim={2cm 0cm 2cm -0.2cm},clip]{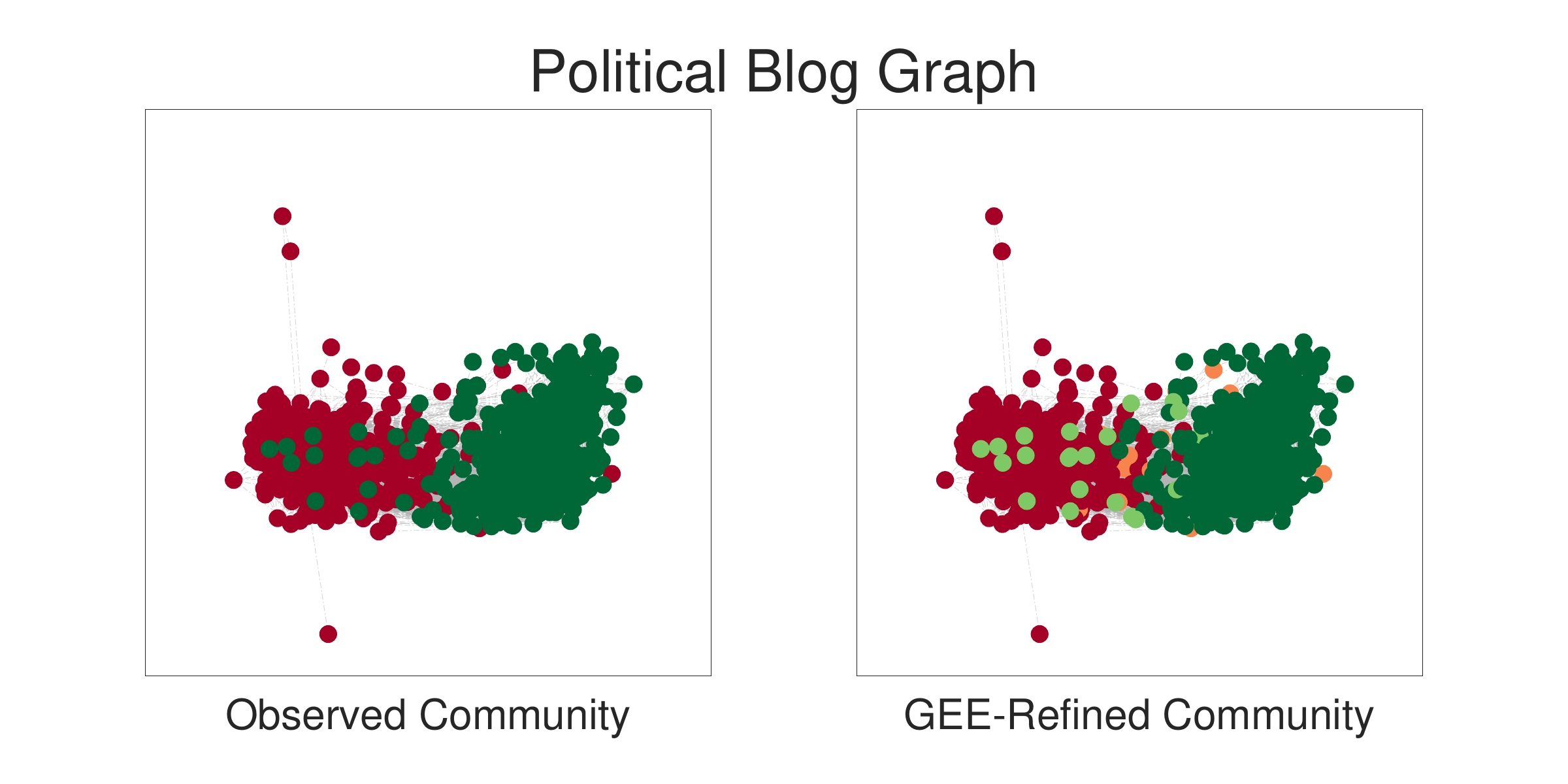}
\caption{This figure visualizes two real graphs, the karate club and political blogs, using observed labels (left panel) and GEE-refined labels after one refinement iteration (right panel). In the left panel, the graphs are drawn with vertices colored by observed labels: dark green and red. In the right panel, the same graphs are shown with vertex colors representing refined classes. For the karate club graph, two vertices from the dark green community are refined into a new group, colored light green. For the political blogs graph, some vertices in the dark green group are refined into a new group, also colored light green, while some vertices in the red group are refined into another new group, colored orange.}
\label{fig3}
\end{figure}

\subsection{Vertex Classification on Multiplex Graphs}
\label{appD}
Some of the real data, specifically the C. elegans data and the Wikipedia data, come with multiple graphs of a common vertex set. The graph encoder embedding can be directly used for multiple graph inputs \cite{GEEFusion} by concatenating the embeddings, as can the refined version. LDA classifier can then be applied to measure the quality of the joint embedding via vertex classification.

For the spectral embedding, we use the unfolded spectral embedding \cite{Patrick2022a}: Given $M$ graphs of matched common vertices, the unfolded version concatenates all adjacency matrices by rows into $\mathbf{A} \in \mathbb{R}^{n \times Mn}$, and applies SVD to yield $\mathbf{Z}^{UASE}=\mathbf{V}_{d}\mathbf{S}_{d}^{0.5} \in \mathbb{R}^{Mn \times d}$, where each $n \times d$ matrix is the embedding for the corresponding graph. We then reshape and concatenate the embedding into $\mathbb{R}^{n \times Md}$ and carry out the vertex classification using linear discriminant analysis.

For node2vec, we simply apply node2vec to each graph, concatenate their embeddings, and apply LDA. Everything else is exactly the same as in Section~\ref{secReal}, and Table~\ref{table2} reports the average vertex classification error and the standard deviation for different combinations of the graph data. The results are consistent with those in Table~\ref{table1}, where refined GEE always improves over the original GEE, and is the best performer throughout all combinations. Additionally, using multiple graphs improves over single graph results.

\begin{table}
\renewcommand{\arraystretch}{1.3}
\centering
\small
{\begin{tabular}{|c||c|c|c|c|c|}
 \hline
 & R-GEE & GEE & U-ASE & U-LSE & N2V\\
\hline
C-Elegans Ac+Ag & $\textbf{33.7} \pm 2.2$ & $40.2 \pm 1.7$  & $34.7 \pm 1.3$ & $40.2 \pm 1.3$ & $51.9 \pm 2.8$\\
Wiki TE+TF& $\textbf{14.6} \pm 0.5$ & $18.0 \pm 0.3$ & $20.7 \pm 0.3$ & $21.1 \pm 0.3$ & n/a\\ 
Wiki TE+GE & $\textbf{14.6} \pm 0.6$ & $17.8 \pm 0.4$ & $21.2 \pm 0.3$ & $30.2 \pm 0.3$  & n/a\\
Wiki TF+GF & $\textbf{15.7} \pm 0.5$ & $18.7 \pm 0.3$ & $21.1 \pm 0.3$ & $31.3 \pm 0.4$  & n/a\\
Wiki GE+GF & $\textbf{32.2} \pm 0.8$ & $39.2 \pm 0.8$ & $43.7 \pm 0.3$ & $50.8 \pm 0.5$  & $39.9 \pm 0.7$\\
Wiki TE+TF+GE+GF & $\textbf{13.3} \pm 0.5$ & $16.1 \pm 0.3$ & $18.0 \pm 0.3$ & $27.9 \pm 0.4$ & n/a\\ 
\hline
\end{tabular}
\caption{This table reports the vertex classification results for multiple-graph data with a common vertex set. All numbers are in percentile. }
\label{table2}
}
\end{table}

\section{Conclusion}

This paper introduces a refined graph encoder embedding method, offering both theoretical justification for its effectiveness and practical insights into the role of latent community structures in improving vertex classification. We demonstrate the advantages of the method through a combination of theoretical results, simulations, and real-world experiments.

In contrast to deep learning-based methods such as GCNs and their many variants, our approach offers unique advantages. Deep learning methods excel at numerical optimization and adapt well to various graph-based tasks by defining appropriate objective functions and utilizing gradient descent for parameter optimization. However, these methods often lack theoretical foundation, and their black-box nature provides limited interpretability regarding the graph's structure. They can also be computationally slow for large graphs. By contrast, the refined graph encoder embedding method presented in this paper not only delivers excellent performance but is also theoretically grounded and, most importantly, provides valuable insights into the latent community structures of the graph --- an area where deep learning methods are inherently limited.

Looking ahead, several promising directions for future research emerge. While this study employs linear discriminant analysis for self-training, alternative approaches, including deep learning-based classifiers, could be investigated to further enhance empirical performance on real-world datasets. Additionally, the discovery of underlying latent communities opens avenues for a wide range of downstream analyses, such as improving semi-supervised learning, advancing graph recommendation systems, and addressing other graph-based tasks. Another intriguing possibility is to use the graph encoder embedding as a structured node feature initialization for graph neural networks. Recent work has demonstrated the effectiveness of randomized node feature initialization \cite{eliasof2023graph}, suggesting that higher-quality feature representations generated could yield even better results. Lastly, while our study primarily focused on the discovery and utilization of latent communities for vertex classification, their practical value in domain-specific graph data warrants further exploration. This could provide significant insights and enhance the interpretability of community structures across various real graphs.

\section*{Declarations}

\begin{itemize}
\item Availability of data and materials: The datasets generated and/or analysed during the current study are available in the Github GraphEmd repository, [https://github.com/cshen6/GraphEmd].
\item Competing interests: The authors declare that they have no competing interests.
\item Acknowledgements: This work was supported the National Science Foundation DMS-2113099, and by funding from Microsoft Research. 
\end{itemize}

\bibliographystyle{ieeetr}
\bibliography{shen,general}

\newpage 
\appendix

\bigskip
\begin{center}
{\large\bf APPENDIX}
\end{center}

\section{PseudoCode}

\begin{algorithm}
\caption{GEE Self Training via Linear Discriminant Analysis (GEELDA)}
\label{alg1}
\begin{algorithmic}
\Require The graph adjacency matrix $\mathbf{A} \in \mathbf{R}^{n \times n}$ and a label vector $\mathbf{Y} \in \{0,1,\ldots,K\}^{n}$, where $1$ to $K$ represent known labels, and $0$ is a dummy category for vertices with unknown labels. 
\Ensure A linear transformed encoder embedding $\mathbf{Z}_{1} \in \mathbb{R}^{n \times K}$, self-trained new label $\mathbf{Y}_{1} \in \{0,1,\ldots,K\}^{n}$, and an indicator vector $\delta \in [0,1]^{n}$. 
\Function{GEELDA}{$\mathbf{A},\mathbf{Y}$}
\State $\mathbf{Z}=\mbox{GEE}(\mathbf{A},\mathbf{Y})$; \Comment{original one-hot graph encoder embedding}
\State $\mathbf{Z}_{1}=\mbox{LDA}(\mathbf{Z},\mathbf{Y})$; \Comment{transform the encoder embedding by Equation~\ref{eq1}}
\State $[,\mathbf{Y}_{1}]=\mbox{rowmax}(\mathbf{Z}_{1})$; \Comment{the maximum dimension per vertex}
\State $ind=\operatorname{find}(\mathbf{Y}==0)$;  
\State $\mathbf{Y}_{1}(ind)=0$; \Comment{omit vertices with unknown labels}
\State $\delta_1=(\mathbf{Y}\neq \mathbf{Y}_{1})$; 
\EndFunction
\end{algorithmic}
\end{algorithm}

\begin{algorithm}
\caption{Refined Graph Encoder Embedding (R-GEE)}
\label{alg2}
\begin{algorithmic}
\Require The graph adjacency matrix $\mathbf{A} \in \mathbf{R}^{n \times n}$ and a label vector $\mathbf{Y} \in \{0,1,\ldots,K\}^{n}$; number of refinement $\gamma_{K}$ and $\gamma_{Y}$, set to $5$ by default; stopping criterion $\epsilon \in [0,1]$ and $\epsilon_n \in \mathbb{N}$, set to $0.3$ and $5$ by default. 
\Ensure The refined graph encoder embedding $\mathbf{Z} \in \mathbb{R}^{n \times d}$, a concatenated label matrix $\mathbf{Y}$.
\Function{R-GEE}{$\mathbf{A},\mathbf{Y}, \gamma_{K}, \gamma_{Y}, \epsilon, \epsilon_{n}$}
\State $[\mathbf{Z}_1,\mathbf{Y}_{1},\delta_1]=\mbox{GEELDA}(\mathbf{A},\mathbf{Y})$; $\mathbf{Y}=\mathbf{Y}_1$;
\For{$k=1,\ldots,\gamma_{Y}$}
\State $[\mathbf{Z}_{2},\mathbf{Y}_{2},\delta_2]=\mbox{GEELDA}(\mathbf{A},\mathbf{Y}_{1})$;
\If{$sum(\delta_1) - max( sum(\delta_1) * \epsilon, \epsilon_n) < sum(\delta_1 \cdot \delta_2)$}
\State Break;
\Else
\State $\mathbf{Z}_1=[\mathbf{Z}_1,\mathbf{Z}_{2}]$; $\mathbf{Y}_{1} = \mathbf{Y}_{2}$; $\mathbf{Y}=[\mathbf{Y},\mathbf{Y}_1]$;
\State $\delta_1 = \delta_1 \cdot \delta_2$;
\EndIf
\EndFor
\For{$k=1,\ldots,\gamma_{K}$}
\State $[\mathbf{Z}_{2},\mathbf{Y}_{2},\delta_2]=\mbox{GEELDA}(\mathbf{A},\mathbf{Y}_{1}+ \delta_1*K)$;
\If{$sum(\delta_1) - max( sum(\delta_1) * \epsilon, \epsilon_n) < sum(\delta_1 \cdot \delta_2)$}
\State Break;
\Else
\State $\mathbf{Z}_1=[\mathbf{Z}_1,\mathbf{Z}_{2}]$; $\mathbf{Y}_{1} = \mathbf{Y}_{2}$; $\mathbf{Y}=[\mathbf{Y},\mathbf{Y}_1]$;
\State $\delta_1 = \delta_1 \cdot \delta_2$;
\EndIf
\EndFor
\EndFunction
\end{algorithmic}
\end{algorithm}

\section{Theorem Proofs}
\clt*
\begin{proof}

First, a necessary assumption is that as $n$ goes to infinity, so does $n_k$; i.e., as the number of vertices goes to infinity, the number of vertices per class also increases to infinity. This is a standard regularity assumption in pattern recognition because, without it, the class would become trivial as $n$ increases.

Under SBM, each dimension $k=1,\ldots,K$ of the vertex embedding satisfies
\begin{align*}
\mathbf{Z}(i,k) & = \mathbf{A}(i,:)\mathbf{W}(:,k) \\
                & = \frac{\sum_{j=1}^{n} I(\mathbf{Y}(j)=k)\mathbf{A}(i,j)}{n_k}\\
                & = \frac{\sum_{j=1, j\neq i, \mathbf{Y}(j)=k}^{n} Bern(\mathbf{B}(y,k))}{n_k}.
\end{align*}
If $k=y$, the numerator is a summation of $(n_k-1)$ i.i.d. Bernoulli random variables, since the summation includes a diagonal entry of $\mathbf{A}$, which is always $0$. Otherwise, $k \neq y$ and the numerator is a summation of $n_k$ i.i.d. Bernoulli random variables. 

Checking the Lyapunov condition and applying the central limit theorem, we have
\begin{align*}
\sqrt{n_k}(\mathbf{Z}(i,k) - \mathbf{B}(y,k)) \stackrel{d}{\rightarrow} \mathcal{N}(0,\mathbf{B}(y,k)(1-\mathbf{B}(y,k))). 
\end{align*}
for each dimension $k$.

Note that $\mathbf{Z}(i,k)$ and $\mathbf{Z}(i,l)$ are always independent when $k \neq l$. This is because every vertex belongs to a unique class, so the same Bernoulli random variable never appears in another dimension. Concatenating every dimension yields that
\begin{align*}
Diag(\vec{n})^{0.5} \cdot (\mathbf{Z}(i,:) - \mathbf{B}(y,:)) \stackrel{d}{\rightarrow} \mathcal{N}(0,\Sigma_{y}).
\end{align*}
Note that more details on the Lyapunov condition, as well as cases for other random graph models, can be found in Theorem 1 in \cite{GEE1}.

Now, given $\mathbf{Z}(i,:)$ is normally distributed for $n$ large, it follows immediately from classical pattern recognition \cite{DevroyeGyorfiLugosiBook} that under the normality assumption and a common variance across all $k$, the linear transformation in Equation~\ref{eq1} estimates the conditional probability. This is because the LDA transformation directly estimates $Prob(Y|X)$ when $X|Y$ is normally distributed. Specifically,
\begin{align*}
\mathbf{Z}_{1}(i,k) = \mathbf{Z}(i,k) \Sigma^{+} \mu_{k}- ((\mu_{k}' \Sigma^{+} \mu_{k}) - \mbox{log}(n_K / n))
\end{align*}
is the exact LDA transformation for each class $k = 1,\ldots,K$. Writing it into a matrix expression for all $k$ leads to Equation~\ref{eq1}.
\end{proof}

\refine*
\begin{proof}
From Theorem~\ref{thm1}, it is immediate that the encoder embedding satisfies the law of large numbers, such that
\begin{align*}
\| \mathbf{Z}(i,:) -\mathbf{B}(\mathbf{Y}(i),:) \|_2 \stackrel{n}{\rightarrow} 0.
\end{align*}
It follows that
\begin{align*}
&  \| \mathbf{Z}(i,:) - \mathbf{Z}(j,:)\|_2 - \|(\mathbf{B}(\mathbf{Y}(i),:)-\mathbf{B}(\mathbf{Y}(j),:)) \|_2 \\
\leq & \| \mathbf{Z}(i,:) - \mathbf{Z}(j,:) - (\mathbf{B}(\mathbf{Y}(i),:)-\mathbf{B}(\mathbf{Y}(j),:)) \|_2 \\
= & \| (\mathbf{Z}(i,:) - \mathbf{B}(\mathbf{Y}(i),:)) - (\mathbf{Z}(j,:)-\mathbf{B}(\mathbf{Y}(j),:)) \|_2 \\
\leq & \| \mathbf{Z}(i,:) - \mathbf{B}(\mathbf{Y}(i),:)\|_{2} + \| \mathbf{Z}(i,:) -\mathbf{B}(\mathbf{Y}(j),:)\|_{2} \\
\rightarrow & 0.
\end{align*}
Since the graph encoder embedding is fully dependent on the given labels, when the latent labels are used, we also have
\begin{align*}
\| \mathbf{Z}_0(i,:) -\mathbf{B}(\mathbf{Y}_{0}(i),:) \|_2 \stackrel{n}{\rightarrow} 0,
\end{align*}
so the same derivation and convergence apply to the encoder embedding using latent labels as well. 
\end{proof}

\section{Additional Experiments}

Generally, we observe that varying parameters can lead to differences in results, but for most simulations and real datasets, the impact is minimal. Here, we examine parameter sensitivity, specifically for $\epsilon$ and $\epsilon_n$, under four settings:
\begin{itemize}
\item Setting 1: $\epsilon= 0.6$ and $\epsilon_n=50$;
\item Setting 2: $\epsilon= 0.4$ and $\epsilon_n=20$; 
\item Setting 3: $\epsilon= 0.2$ and $\epsilon_n=5$;
\item Setting 4: $\epsilon= 0.02$ and $\epsilon_n=2$; 
\end{itemize}
The refinement becomes progressively more aggressive across these settings. For example, consider Equation~\ref{close} with $\sum_{i=1}^{n} \delta_1(i) = 100$ (indicating 100 mismatched vertices based on the LDA self-learned labels versus the previous labels) and $\sum_{i=1}^{n} \delta_1(i) \cdot \delta_2(i) = 70$ (where 70 mismatches remain after refinement). In this case, the difference is $30$. Under Settings 1 and 2, refinement would stop, as $30$ is smaller than $\max(\sum_{i=1}^{n} \delta_1(i) * \epsilon, \epsilon_n)$. However, under Settings 3 and 4, refinement would continue, as the difference exceeds the thresholds for these more aggressive settings.

Figure~\ref{fig7} illustrates the vertex classification error and the precision of recovering latent communities under the four parameter settings for simulated graph 2 in Figure~\ref{fig2}. In the left panel of Figure~\ref{fig7}, we see that precision is only slightly affected by parameter changes, decreasing slightly as refinement becomes more aggressive. This indicates that a few vertices are incorrectly reassigned to new communities due to overly aggressive refinement. However, as shown in the right panel, this has a negligible effect on vertex classification error, where all four settings yield nearly identical results.

\begin{figure}
\centering
\includegraphics[width=0.9\linewidth,trim={0cm 0cm 0cm 0cm},clip]{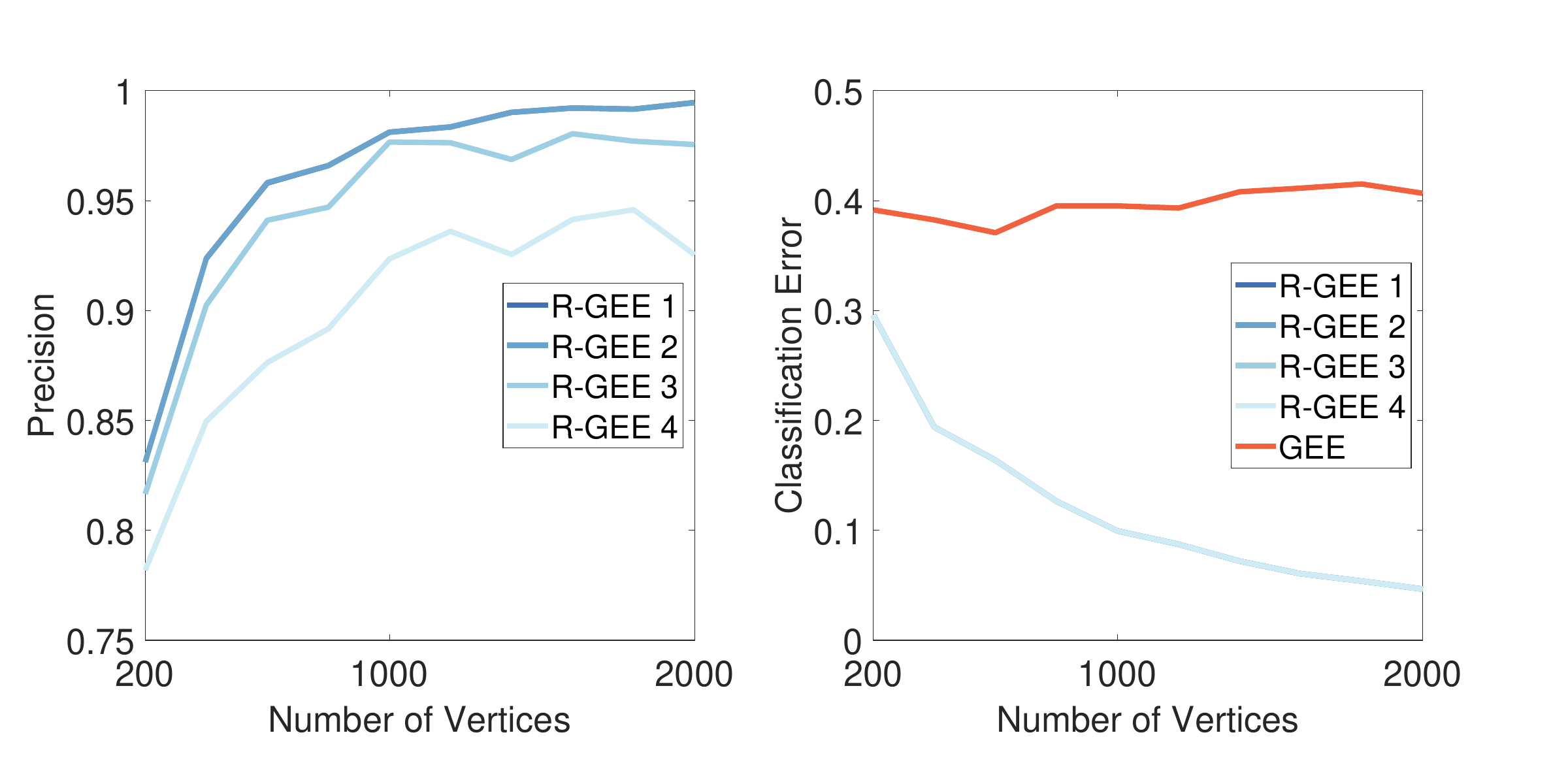}
\caption{The left panel displays the precision of the refined GEE in recovering latent communities under four different parameter settings. The right panel shows the vertex classification error.}
\label{fig7}
\end{figure}

\end{document}